\documentclass[preprint,12pt,authoryear]{elsarticle}

\usepackage{amssymb}
\usepackage{amsmath}
\usepackage{amsthm}
\newtheorem{definition}{Definition}
\newtheorem{assumption}{Assumption}
\newtheorem{lemma}{Lemma}
\newtheorem{corollary}{Corollary}
\newtheorem{example}{Example}
\newtheorem{remark}{Remark}

\usepackage{color}

\journal{arXiv}

\begin{document}

\begin{frontmatter}



\title{Rescaled Bayes factors: a class of e-variables} 


\author{Thorsten Dickhaus} 

\affiliation{organization={Institute for Statistics, University of Bremen},
            addressline={P. O. Box 330 440}, 
            city={Bremen},
            postcode={28344}, 
            state={Bremen},
            country={Germany}}

\begin{abstract}
A class of e-variables is introduced and analyzed. Some examples are presented.
\end{abstract}



\begin{keyword}
Frequentist-Bayes reconciliation \sep hypothesis testing \sep likelihood ratio order \sep safe anytime-valid inference


\MSC[2020] 62F03 
\end{keyword}

\end{frontmatter}



\section{Introduction}\label{sec1}
During recent years, safe anytime-valid inference (SAVI) has emerged as a hot topic in theoretical statistics; cf. \cite{Ramdas-StatScience} for an overview. Due to Ville’s inequality (see, e. g., \cite{Choi1988}), it is possible to test with so-called e-values and their running product the same null hypothesis several times (with an increasing amount of data) without compromising type I error control. In this, the defining property of an e-variable is that its (conditional) expectation under the null hypothesis is upper-bounded by one; see Section \ref{sec2} for mathematical details. 

In this paper, a class of e-variables is introduced and analyzed. Furthermore, some examples will be presented. The rest of the material is structured as follows. In Section \ref{sec2}, the general concept will be explained. Section \ref{sec3} is concerned with analytical solutions to the problem, while Section \ref{sec4} deals with numerical approximations. In Section \ref{sec5}, a case study is presented. We conclude with a discussion in Section \ref{sec6}. 

\section{The general concept}\label{sec2}
Let $\left(\mathcal{X}, \mathcal{B}(\mathcal{X}), \{\mathbb{P}_\theta: \theta \in \Theta \subseteq \mathbb{R}^d\}\right)$ denote a parametric statistical model, where $\mathcal{X}$ denotes the sample space, $\mathcal{B}(\mathcal{X})$ denotes a suitable $\sigma$-field over $\mathcal{X}$, $\theta$ is the parameter of the model taking its values in the parameter space $\Theta$, and $d \in \mathbb{N}$ is a given integer. Assume that $\Theta$ is equipped with a suitable $\sigma$-field $\mathcal{B}(\Theta)$, let $\Theta_0 \subset \Theta$ denote a measurable, non-empty subset of $\Theta$, and let $\Theta_1 = \Theta \setminus \Theta_0$. We will consider test problems of the type
\begin{equation}\label{testproblem}
H_0: \theta \in \Theta_0 \text{~~versus~~} H_1: \theta \in \Theta_1.
\end{equation}
We will refer to $H_0$ as the null hypothesis (or ``null'' for short) and to $H_1$ as the alternative (hypothesis). 

\begin{definition}
An e-variable for the test problem \eqref{testproblem} is a statistic $E: \left(\mathcal{X}, \mathcal{B}(\mathcal{X})\right) \to \left([0, \infty], \mathcal{B}([0, \infty])\right)$ fulfilling that
\[
\forall \theta_0 \in \Theta_0: \mathbb{E}_{\theta_0}\left[E\right] \leq 1.
\]
Realizations of e-variables are called e-values.
\end{definition}

We will from now on assume that the family $\left(\mathbb{P}_\theta\right)_{\theta \in \Theta}$ is dominated by some reference probability measure $\nu$, and we will denote the density (Radon-Nikodym derivate) of $\mathbb{P}_\theta$ with respect to $\nu$ by $f_\theta$, $\theta \in \Theta$. We will refer to $f_\theta$ as the likelihood (function) of the model under $\theta$.

Let $\pi$ denote a (prior) probability distribution on $\left(\Theta, \mathcal{B}(\Theta)\right)$. Then, the prior probability of $\Theta_i$ is given by 
\[
\pi_i \equiv \mathbb{P}\left(\Theta_i\right) = \int_{\Theta_i} \pi(d \theta) = \int_{\Theta_i} \pi(\theta) d\theta, \; i \in \{0, 1\}, 
\]
where we identify the measure $\pi$ with its Lebesgue density, assuming the existence of the latter. Throughout the remainder, we assume that $\min\{\pi_0, \pi_1\} > 0$, to avoid pathologies.

\begin{definition}
Let $\mathbf{x}$ be given data, interpreted as the realization of a random variate $\mathbf{X}$ which formalizes the data-generating process. Then, the Bayes factor BF (in favor of $H_1$ as compared to $H_0$) is given by
\begin{equation}\label{BF-defi}
	\mathrm{BF}\left(\mathbf{x}\right) 
	= 
	\frac{\int_{\Theta_{1}} f_\theta(\mathbf{x}) \pi(d \theta)  / \pi_1}
		 {\int_{\Theta_{0}} f_\theta(\mathbf{x}) \pi(d \theta)  / \pi_0} 
	= 
	\frac{\mathbb{P}\left(\Theta_{1} | \mathbf{X} = \mathbf{x} \right) / \mathbb{P}\left(\Theta_{0} | \mathbf{X} = \mathbf{x}\right)}{\mathbb{P}(\Theta_{1}) / \mathbb{P}(\Theta_{0})},
\end{equation}
where $\mathbb{P}$ refers to the joint distribution of data and parameter.
\end{definition}
Bayes factors are the basis of many Bayesian tests; cf., e.\ g., Section 4.1.4 in \cite{Lee-Bayes4} or Section 8.3.1 in \cite{Spok-Dickhaus-Buch}. Despite these origins in Bayesian statistics, we may interpret $\mathrm{BF}\left(\mathbf{x}\right)$ simply as a summary of the data $\mathbf{x}$, thus as a statistic with values in $[0, \infty]$. This interpretation has been indicated by \cite{Yekutieli-TEST} and by \cite{Dickhaus-Bayes}, among others. In particular, we are with this interpretation able to compute the expected value of the random variable $\mathrm{BF}\left(\mathbf{X}\right)$ under a given $\theta \in \Theta$ as follows:
\begin{equation}\label{BF-expectation}
\mathbb{E}_{\theta}\left[\mathrm{BF}\right]=\int_{\mathcal{X}}\mathrm{BF}(\mathbf{x})f_{\theta}(\mathbf{x}) \nu(d \mathbf{x}).
\end{equation}
Now, consider the following general assumption.
\begin{assumption}\label{general-assumption}
For a given model, a given test problem, and a given prior, the maximum 
\begin{equation}\label{mu-star}
\mu^* = \max_{\theta_0 \in \Theta_0} \mathbb{E}_{\theta_0}\left[\mathrm{BF}\right]
\end{equation}
exists in $(0, \infty)$.
\end{assumption}

\begin{corollary}\label{main-idea}
Under Assumption \ref{general-assumption}, an e-variable $\mathrm{BF}_\text{rescaled}$ for testing $H_0$ versus $H_1$ is given by 
\[
\mathrm{BF}_\text{rescaled}(\mathbf{x}) = \mathrm{BF}\left(\mathbf{x}\right) / \mu^*.
\]
We call $\mathrm{BF}_\text{rescaled}$ the rescaled version of the Bayes factor $\mathrm{BF}$.
\end{corollary}

\begin{remark}
The quantity $\mu^*$ from \eqref{mu-star} is of interest in its own right. Namely, if $\mu^* > 1$, then $\mu^*$ quantifies the ``price to pay for SAVI'' on the (relative) Bayes factor scale. 
\end{remark}

\section{Analytical solutions}\label{sec3}
In this section, we present examples in which $\mu^*$ from \eqref{mu-star} can be computed explicitly or in which the parameter value maximizing the right-hand side of \eqref{mu-star} can be determined analytically, respectively.
\begin{example}[Simple null hypothesis]\label{example-simple}
Assume that $\Theta_0 = \{\theta_0\}$ is a singleton. Then, $\mu^* = 1$, and $\mathrm{BF} = \mathrm{BF}_\text{rescaled}$ is an e-variable by itself. This is a direct consequence of Fubini's theorem, since
\begin{eqnarray*}
\mathbb{E}_{\theta_0}\left[\mathrm{BF}\right] &=& \int_{\mathcal{X}} \left[\int_{\Theta_1} f_{\theta}(\mathbf{x}) \pi(d \theta) / \pi_1\right] \nu(d \mathbf{x})\\
&=&\int_{\Theta_1}\left[\int_{\mathcal{X}} f_{\theta}(\mathbf{x}) \nu(d \mathbf{x})\right] \frac{1}{\pi_1} \pi(d \theta) = \frac{\pi_1}{\pi_1} = 1;
\end{eqnarray*}
cf. the third main interpretation of an e-variable provided in Section 1.1 of \cite{safe-testing}.
\end{example}

\begin{example}[Monotonic likelihoods]\label{example-hoang1}
Here, we consider the special case of $\Theta = \mathbb{R}$, and we consider one-sided test problems with $\Theta_0 = (-\infty, 0]$. In this setting, \cite{Hoang-Dickhaus-JSCS} considered cases in which the Bayes factor itself is defined in terms of a (real-valued) summary statistic $S: \left(\mathcal{X}, \mathcal{B}(\mathcal{X})\right) \to \left(\mathbb{R}, \mathcal{B}(\mathbb{R})\right)$; cf. \cite{Johnson-BF-statistics,Johnson-BF-statistics-props}. Letting $g_\theta$ denote the density of $S$ under $\theta$, the Bayes factor based on $S$ is given by
\begin{equation}\label{BF-statistic}
\mathrm{BF}\left(s\right) 
	= 
	\frac{\int_{\Theta_{1}} g_\theta(s) \pi(d \theta)  / \pi_1}
		 {\int_{\Theta_{0}} g_\theta(s) \pi(d \theta)  / \pi_0}, 
\end{equation}
where $s \in \mathbb{R}$ denotes a realization of $S$.

Now assume that the ``likelihood functions'' $\left\{g_\theta: \theta \in \Theta\right\}$ are likelihood ratio-ordered in the sense that
\begin{align}
\forall \theta_0 \in \Theta_0: &~~g_\theta(s) \text{~~is increasing in s},\label{mono1}\\
\forall \theta_1 \in \Theta_1: &~~g_\theta(s) \text{~~is decreasing in s}.\label{mono2}  
\end{align}
Under these assumptions, \cite{Hoang-Dickhaus-JSCS} have shown in their Lemma A.2 that $\mu^* = \mathbb{E}_{0}\left[\mathrm{BF}\right]$.
\end{example}

\begin{example}[Reduction to a singleton]\label{example-hoang2}
Consider the same setting as in Example \ref{example-hoang1}, and assume that the model is such that the monotonicity assumptions \eqref{mono1} and \eqref{mono2} are fulfilled. Then, \cite{Hoang-Dickhaus-JSCS} have also argued that the ``reduced Bayes factor'' $\mathrm{BF}_0$, given by
\begin{equation}\label{reduced-BF}
\mathrm{BF}_0\left(s\right) 
	= 
	\frac{\int_{\Theta_{1}} g_\theta(s) \pi(d \theta)  / \pi_1}{g_0(s)} 
\end{equation}
is an e-variable for the one-sided test problem defined in Example \ref{example-hoang1}. In \eqref{reduced-BF}, the denominator of the Bayes factor based on $S$ is ``reduced'' to a one-point prior in zero under the null. We mention this reduction technique here, although it does not make use of 
re-scaling. 
\end{example}

\begin{remark}
In particular, \cite{Hoang-Dickhaus-JSCS} considered $p$-variables as such summary statistics $S$. Even though the monotonicity assumptions \eqref{mono1} and \eqref{mono2} appear very restrictive, several $p$-value models that are popular in the literature fulfill these assumptions. See Section 4.1 in \cite{Hoang-Dickhaus-JSCS} for concrete examples of such models. In Section \ref{sec5}, we will take up one of the models from Section 4.1 in \cite{Hoang-Dickhaus-JSCS}, for exemplary purposes.
\end{remark}

\section{Numerical approximations}\label{sec4}
In this section, we consider situations in which no analytical solution of the maximization problem from \eqref{mu-star} is available or feasible. We describe numerical approximations of $\mu^*$ under such settings. 
\begin{example}[Monte Carlo simulations]\label{example-monte-carlo}
Let $MC \in \mathbb{N}$ be a given number of Monte Carlo repetitions. As described, for instance, in Section 6 of \cite{Dickhaus-Bayes}, it is possible to draw pseudo-random realizations $\mathbf{x}^{(1)}, \ldots, \mathbf{x}^{(\text{MC})}$ from the sampling distribution  $\mathbb{P}_\theta$ of the data for a given parameter value, to evaluate the Bayes factor for each of these pseudo-random realizations, and to approximate $\mathbb{E}_{\theta}\left[\mathrm{BF}\right]$ by the arithmetic mean of $\left\{\mathrm{BF}\left(\mathbf{x}^{(1)}\right), \ldots, \mathrm{BF}\left(\mathbf{x}^{(\text{MC})}\right)\right\}$. An analogous technique has also been used in Appendix S3 of the Supplementary Material of \cite{Turner-JSPI2024} in order to demonstrate that a certain family of Bayes factors does not possess the e-variable property. This strategy is especially effective if $\Theta_0$ is a finite set, because in this case the simulated values of $\left\{\mathbb{E}_{\theta_0}\left[\mathrm{BF}\right]: \theta_0 \in \Theta_0\right\}$ can simply be tabulated and their maximum can be read off from that table. Examples of statistical models with discrete parameter spaces have been discussed, e.\ g., by \cite{Choirat-discrete} and by \cite{Vajda-finite}.

To provide a concrete example, it may be of interest to test a null hypothesis about a population total under a capture-recapture sampling scheme. In this context, simulation techniques related to the hypergeometric distribution have been proposed in Section 6.1 of \cite{Ross-online-book}.
\end{example}

\begin{example}[Numerical maximization techniques]\label{example-newton}
Especially in cases where $\Theta_0$ is an interval or a Cartesian product of intervals, respectively, numerical maximization techniques like, for instance, a simple line search, or the Newton-Raphson algorithm, or related methods (cf., e.\ g., Part I of \cite{numopt-book}), can be used to find $\mu^*$. In Section \ref{sec5} we will illustrate this technique under a concrete statistical model.
\end{example}

\section{Case study: Beta-distributed $p$-value}\label{sec5}
In this section, we illustrate the concepts discussed in the previous sections by means of a concrete statistical model. Namely, we assume that a $p$-variable $P$ is available as the summary statistic of the data. Following Section 4.1 in \cite{Hoang-Dickhaus-JSCS}, we assume the following model for the Lebesgue density $g_\theta$ of $P$, where $\theta$ is a real-valued parameter:
\begin{equation}\label{p-variable-density}
g_\theta = \begin{cases}
\textrm{Lebesgue density of Beta}\left(1 - \theta, 1\right), &\theta \leq 0,\\
\textrm{Lebesgue density of Beta}\left(1, 1 + \theta\right), &\theta > 0.
\end{cases}
\end{equation}
In particular, $P$ is uniformly distributed on $[0, 1]$ under $\theta = 0$. Realizations of $P$ will be denoted by $p$. As in Examples \ref{example-hoang1} and \ref{example-hoang2}, we consider $\Theta_0 = (-\infty, 0]$ here, too.

\begin{lemma}
The family of densities $\left\{g_\theta: \theta \in \mathbb{R}\right\}$ fulfills the monotonicity conditions \eqref{mono1} and \eqref{mono2} on its (common) support $[0, 1]$. 
\end{lemma}

\begin{proof}
Under $\theta = 0$, $g_0$ is constant on $[0, 1]$. In the case of $\theta < 0$, $g_\theta(p)$ is proportional to $p^{-\theta}$, which clearly is increasing in $p \in [0, 1]$. In the case of $\theta > 0$, $g_\theta(p)$ is proportional to $(1 - p)^{\theta}$, which clearly is decreasing in $p \in [0, 1]$.
\end{proof}

As prior distribution on $\Theta = \mathbb{R}$, we choose here Student's $t$-distribution with five degrees of freedom, $t_5$ for short. The $t_5$-density is reasonably smooth for a proper numerical handling, but it exhibits heavy tails, such that also parameter values with quite a large distance to zero play a non-negligible role in the statistical model. Due to the symmetry of the $t_5$-density, $\pi_0 = \pi_1 = 1/2$ here.

As argued in Example \ref{example-hoang1}, we have here that
\begin{equation}\label{mu-star-beta}
\mu^* = \mathbb{E}_{0}\left[\mathrm{BF}\right] = \int_{0}^1 \frac{\int_{0}^\infty \beta_{1, 1 + \theta}(p) f_{t_5}(\theta) d \theta}
		 {\int_{-\infty}^0 \beta_{1-\theta, 1}(p) f_{t_5}(\theta) d \theta} d p,
\end{equation}
where $\beta_{a, b}$ denotes the Lebesgue density of Beta$(a, b)$.
To obtain (an approximation of) the numerical value of $\mu^*$, we first used the \verb=integrate= function of the \verb=R= software package (Version R x64 3.6.1). This yielded a value of $\mu^* \approx 1.804$. A computer simulation in \verb=R= with $500{,}000$ Monte Carlo repetitions yielded an approximation of $\mu^* \approx 1.807$. These two numerical values agree quite well.

Table \ref{table1} provides simulated values of $\mathbb{E}_{\theta}\left[\mathrm{BF}\right]$ for some other values of $\theta \in \Theta_0$.

\begin{table}[htp]
\centering
\begin{tabular}{|l| c c c c c c|}
\hline
$\theta$                                        &$-2.5$             &$-2.0$             &$-1.5$             
&$-1.0$             &$-0.5$             &$-0.25$\\ 
\hline
$\mathbb{E}_{\theta}\left[\mathrm{BF}\right]$   &$\approx 0.461$    &$\approx 0.533$    &$\approx 0.641$    
&$\approx 0.814$    &$\approx 1.122$   &$\approx 1.386$\\ 
\hline
\end{tabular}
\caption{Simulated expected values of the Bayes factor for the model discussed in Section \ref{sec5}, for several parameter values in the null hypothesis. Each simulated value is based on $500{,}000$ Monte Carlo repetitions.}\label{table1}
\end{table}
The monotonicity property of $\theta \mapsto \mathbb{E}_{\theta}\left[\mathrm{BF}\right]$ can clearly be read off from the values tabulated in Table \ref{table1}. It is remarkable that $\mathbb{E}_{\theta}\left[\mathrm{BF}\right]$ exceeds $1.0$ for all $\theta \geq -0.5$, demonstrating that the rescaling of the Bayes factor, which we have proposed in Section \ref{sec2}, is indeed necessary to guarantee SAVI in this model.

Finally, we verified all values tabulated in Table \ref{table1} by means of the \verb=integrate= function of the \verb=R= software, and we maximized the function $\theta \mapsto \mathbb{E}_{\theta}\left[\mathrm{BF}\right]$ by means of the \verb=optimize= function of the \verb=R= software. This function numerically determined $\theta = 0$ as the parameter value maximizing $\mathbb{E}_{\theta}\left[\mathrm{BF}\right]$, with a numerical approximation of $\mu^* = \mathbb{E}_{0}\left[\mathrm{BF}\right] \approx 1.804$, as before.

All \verb=R= computations described in this section ran very fast. For example, producing all six values tabulated in Table \ref{table1} took altogether less than 15 minutes on a standard laptop computer. The respective \verb=R= worksheets are available from the author upon request.

\section{Discussion}\label{sec6}
Under Assumption \ref{general-assumption}, rescaling Bayes factors is a general way to transform Bayes factors into e-variables. Following \cite{Dickhaus-RSS-dicussion}, one may call this rescaling operation a ``BF-to-e-calibration'', in the spirit of the ``p-to-e-calibration'' methods discussed in Section 2 of \cite{VovkWang2021-Annals}. In the (somewhat trivial) case of a simple null hypothesis, this calibration has certain optimality properties regarding power, as extensively argued by \cite{safe-testing}. In the case of composite null hypotheses, other calibration methods may be more powerful in certain situations. However, the simplicity of the rescaling operation may nevertheless be convenient in practice.

Of course, it is also possible to rescale other non-negative summary statistics in order to calibrate them to the e-value scale. However, due to the close interrelations of Bayes factors and e-variables worked out in previous literature, we consider Bayes factors a canonical choice for this calibration.



%

\begin{thebibliography}{17}
\expandafter\ifx\csname natexlab\endcsname\relax\def\natexlab#1{#1}\fi
\providecommand{\url}[1]{\texttt{#1}}
\providecommand{\href}[2]{#2}
\providecommand{\path}[1]{#1}
\providecommand{\DOIprefix}{doi:}
\providecommand{\ArXivprefix}{arXiv:}
\providecommand{\URLprefix}{URL: }
\providecommand{\Pubmedprefix}{pmid:}
\providecommand{\doi}[1]{\href{http://dx.doi.org/#1}{\path{#1}}}
\providecommand{\Pubmed}[1]{\href{pmid:#1}{\path{#1}}}
\providecommand{\bibinfo}[2]{#2}
\ifx\xfnm\relax \def\xfnm[#1]{\unskip,\space#1}\fi
\bibitem[{Bonnans et~al.(2006)Bonnans, Gilbert, Lemar\'echal and
  Sagastiz\'abal}]{numopt-book}
\bibinfo{author}{Bonnans, J.F.}, \bibinfo{author}{Gilbert, J.C.},
  \bibinfo{author}{Lemar\'echal, C.}, \bibinfo{author}{Sagastiz\'abal, C.A.},
  \bibinfo{year}{2006}.
\newblock \bibinfo{title}{Numerical optimization. Theoretical and practical
  aspects}.
\newblock Universitext. \bibinfo{edition}{{Second}} ed.,
  \bibinfo{publisher}{Springer-Verlag, Berlin}.
\bibitem[{Choi(1988)}]{Choi1988}
\bibinfo{author}{Choi, K.P.}, \bibinfo{year}{1988}.
\newblock \bibinfo{title}{Some sharp inequalities for martingale transforms}.
\newblock \bibinfo{journal}{Transactions of the American Mathematical Society}
  \bibinfo{volume}{307}, \bibinfo{pages}{279--300}.
\newblock \DOIprefix\doi{10.2307/2000763}.
\bibitem[{Choirat and Seri(2012)}]{Choirat-discrete}
\bibinfo{author}{Choirat, C.}, \bibinfo{author}{Seri, R.},
  \bibinfo{year}{2012}.
\newblock \bibinfo{title}{Estimation in discrete parameter models}.
\newblock \bibinfo{journal}{Statistical Science} \bibinfo{volume}{27},
  \bibinfo{pages}{278--293}.
\newblock \DOIprefix\doi{10.1214/11-STS371}.
\bibitem[{Dickhaus(2015)}]{Dickhaus-Bayes}
\bibinfo{author}{Dickhaus, T.}, \bibinfo{year}{2015}.
\newblock \bibinfo{title}{Simultaneous {B}ayesian analysis of contingency
  tables in genetic association studies}.
\newblock \bibinfo{journal}{Statistical Applications in Genetics and Molecular
  Biology} \bibinfo{volume}{14}, \bibinfo{pages}{347--360}.
\newblock \URLprefix \url{https://doi.org/10.1515/sagmb-2014-0052},
  \DOIprefix\doi{10.1515/sagmb-2014-0052}.
\bibitem[{Dickhaus(2024)}]{Dickhaus-RSS-dicussion}
\bibinfo{author}{Dickhaus, T.}, \bibinfo{year}{2024}.
\newblock \bibinfo{title}{{Thorsten Dickhaus's contribution to the Discussion
  of ``Safe Testing'' by Grünwald, de Heide, and Koolen}}.
\newblock \bibinfo{journal}{Journal of the Royal Statistical Society. Series B.
  Statistical Methodology} \bibinfo{volume}{86}.
\bibitem[{Gr{\"u}nwald et~al.(2024)Gr{\"u}nwald, de~Heide and
  Koolen}]{safe-testing}
\bibinfo{author}{Gr{\"u}nwald, P.}, \bibinfo{author}{de~Heide, R.},
  \bibinfo{author}{Koolen, W.}, \bibinfo{year}{2024}.
\newblock \bibinfo{title}{{{S}afe {T}esting}}.
\newblock \bibinfo{journal}{Journal of the Royal Statistical Society. Series B.
  Statistical Methodology} \bibinfo{volume}{86}.
\bibitem[{Hoang and Dickhaus(2022)}]{Hoang-Dickhaus-JSCS}
\bibinfo{author}{Hoang, A.T.}, \bibinfo{author}{Dickhaus, T.},
  \bibinfo{year}{2022}.
\newblock \bibinfo{title}{Combining independent {$p$}-values in replicability
  analysis: a comparative study}.
\newblock \bibinfo{journal}{Journal of Statistical Computation and Simulation}
  \bibinfo{volume}{92}, \bibinfo{pages}{2184--2204}.
\newblock \URLprefix \url{https://doi.org/10.1080/00949655.2021.2022678},
  \DOIprefix\doi{10.1080/00949655.2021.2022678}.
\bibitem[{Johnson(2005)}]{Johnson-BF-statistics}
\bibinfo{author}{Johnson, V.E.}, \bibinfo{year}{2005}.
\newblock \bibinfo{title}{Bayes factors based on test statistics}.
\newblock \bibinfo{journal}{Journal of the Royal Statistical Society. Series B.
  Statistical Methodology} \bibinfo{volume}{67}, \bibinfo{pages}{689--701}.
\newblock \DOIprefix\doi{10.1111/j.1467-9868.2005.00521.x}.
\bibitem[{Johnson(2008)}]{Johnson-BF-statistics-props}
\bibinfo{author}{Johnson, V.E.}, \bibinfo{year}{2008}.
\newblock \bibinfo{title}{{Properties of Bayes Factors Based on Test
  Statistics}}.
\newblock \bibinfo{journal}{Scandinavian Journal of Statistics}
  \bibinfo{volume}{35}, \bibinfo{pages}{354--368}.
\newblock \URLprefix \url{http://www.jstor.org/stable/41548598}.
\bibitem[{Lee(2012)}]{Lee-Bayes4}
\bibinfo{author}{Lee, P.M.}, \bibinfo{year}{2012}.
\newblock \bibinfo{title}{Bayesian statistics. {An} introduction.}
\newblock \bibinfo{edition}{4th} ed., \bibinfo{publisher}{Chichester: John
  Wiley \& Sons}.
\bibitem[{Ramdas et~al.(2023)Ramdas, Gr{\"u}nwald, Vovk and
  Shafer}]{Ramdas-StatScience}
\bibinfo{author}{Ramdas, A.}, \bibinfo{author}{Gr{\"u}nwald, P.},
  \bibinfo{author}{Vovk, V.}, \bibinfo{author}{Shafer, G.},
  \bibinfo{year}{2023}.
\newblock \bibinfo{title}{Game-theoretic statistics and safe anytime-valid
  inference}.
\newblock \bibinfo{journal}{Statistical Science} \bibinfo{volume}{38},
  \bibinfo{pages}{576--601}.
\newblock \DOIprefix\doi{10.1214/23-STS894}.
\bibitem[{Ross(2022)}]{Ross-online-book}
\bibinfo{author}{Ross, K.}, \bibinfo{year}{2022}.
\newblock \bibinfo{title}{{An Introduction to Probability and Simulation}}.
\newblock \bibinfo{publisher}{bookdown.org}.
\newblock \URLprefix \url{https://bookdown.org/kevin\_davisross/probsim-book/}.
\bibitem[{Spokoiny and Dickhaus(2015)}]{Spok-Dickhaus-Buch}
\bibinfo{author}{Spokoiny, V.}, \bibinfo{author}{Dickhaus, T.},
  \bibinfo{year}{2015}.
\newblock \bibinfo{title}{Basics of modern mathematical statistics}.
\newblock Springer Texts in Statistics, \bibinfo{publisher}{Berlin: Springer}.
\newblock \DOIprefix\doi{10.1007/978-3-642-39909-1}.
\bibitem[{Turner et~al.(2024)Turner, Ly and Gr{\"{u}}nwald}]{Turner-JSPI2024}
\bibinfo{author}{Turner, R.J.}, \bibinfo{author}{Ly, A.},
  \bibinfo{author}{Gr{\"{u}}nwald, P.D.}, \bibinfo{year}{2024}.
\newblock \bibinfo{title}{Generic {E}-variables for exact sequential
  {$k$}-sample tests that allow for optional stopping}.
\newblock \bibinfo{journal}{Journal of Statistical Planning and Inference}
  \bibinfo{volume}{230}, \bibinfo{pages}{Paper No. 106116, 15}.
\newblock \URLprefix \url{https://doi.org/10.1016/j.jspi.2023.106116},
  \DOIprefix\doi{10.1016/j.jspi.2023.106116}.
\bibitem[{Vajda(1967)}]{Vajda-finite}
\bibinfo{author}{Vajda, I.}, \bibinfo{year}{1967}.
\newblock \bibinfo{title}{On the statistical decision problems with finite
  parameter space}.
\newblock \bibinfo{journal}{Kybernetika} \bibinfo{volume}{3},
  \bibinfo{pages}{451--466}.
\bibitem[{Vovk and Wang(2021)}]{VovkWang2021-Annals}
\bibinfo{author}{Vovk, V.}, \bibinfo{author}{Wang, R.}, \bibinfo{year}{2021}.
\newblock \bibinfo{title}{{E-values: Calibration, combination, and
  applications}}.
\newblock \bibinfo{journal}{{Annals of Statistics}} \bibinfo{volume}{49},
  \bibinfo{pages}{1736--1754}.
\bibitem[{Yekutieli(2015)}]{Yekutieli-TEST}
\bibinfo{author}{Yekutieli, D.}, \bibinfo{year}{2015}.
\newblock \bibinfo{title}{Bayesian tests for composite alternative hypotheses
  in cross-tabulated data}.
\newblock \bibinfo{journal}{TEST} \bibinfo{volume}{24},
  \bibinfo{pages}{287--301}.
\newblock \URLprefix \url{https://doi.org/10.1007/s11749-014-0407-1},
  \DOIprefix\doi{10.1007/s11749-014-0407-1}.

\end{thebibliography}
\end{document}